%% file: vldb-main.tex
\documentclass[sigconf, nonacm]{acmart}

\usepackage{xspace}
\usepackage{multirow}
\usepackage{amsmath}
\usepackage{algorithm}
\usepackage{subfig}

\usepackage[noend]{algpseudocode}

\newtheorem{lemma}{Lemma}
\newtheorem{theorem}{Theorem}

\newcommand{\Sys}{\textsc{AsymCache}\xspace}

\DeclareMathOperator*{\argmin}{argmin}

\makeatletter
\algnewcommand{\LineComment}[1]{\Statex \hskip\ALG@thistlm \(\triangleright\) #1}
\makeatother

\begin{document}
\title{Multi-Segment Attention: Enabling Efficient KV-Cache Management for Faster Large Language Model Serving}

\author{Chunan Shi}
\affiliation{%
  \institution{Peking University}
  \city{Beijing}
  \country{China}
}
\email{spirited\_away@pku.edu.cn}

\author{Yilei Chen}
\affiliation{%
  \institution{Peking University}
  \city{Beijing}
  \country{China}
}
\email{chenyilei@stu.pku.edu.cn}

\author{Yilin Chen}
\affiliation{%
  \institution{Peking University}
  \city{Beijing}
  \country{China}
}
\email{yilinchen25@stu.pku.edu.cn}

\author{Xupeng Miao}\authornote{Corresponding authors.}
\affiliation{%
  \institution{Peking University}
  \city{Beijing}
  \country{China}
}
\email{xupeng.miao@pku.edu.cn}

\author{Bin Cui}\authornotemark[1]
\affiliation{%
  \institution{Peking University}
  \city{Beijing}
  \country{China}
}
\email{bin.cui@pku.edu.cn}

\input{chapters/abstract}

\maketitle

\input{chapters/introduction}

\input{chapters/background}
\input{chapters/design}
\input{chapters/scheduler-new}
\input{chapters/evaluation}
\input{chapters/conclusion}


\bibliographystyle{ACM-Reference-Format}
\bibliography{sample}

\clearpage
\appendix
\input{chapters/appendix}

\end{document}

%% file: chapters/abstract.tex
\begin{abstract}

Large Language Model (LLM) inference relies on key–value (KV) caches to avoid redundant attention computation.
While approximate KV cache retention techniques reduce memory usage by sacrificing model accuracy, lossless approaches instead evict KV cache blocks from GPU memory and reconstruct them on demand to preserve exact outputs.
Existing lossless KV cache management systems primarily base eviction decisions on access frequency or positional heuristics, without considering how different KV cache blocks affect the execution efficiency of GPU attention kernels.

In this paper, we propose \Sys, a computation-latency-aware KV cache management system for LLM inference that explicitly aligns cache residency decisions with GPU attention kernel performance, including three key components: 
Multi‑Segment Attention (MSA) for efficient non‑contiguous KV context processing, a cache eviction policy that jointly optimizes hit rate and position‑aware recomputation cost, and an adaptive chunking scheduler for high hardware utilization.
Experiments show that \Sys reduces TTFT by up to 1.90–2.03× and time-per-output-token (TPOT) by 1.62–1.71× over latest baselines, confirming the effectiveness of the method in common workloads and validating its design goal of balancing computational efficiency with cache hit rate.
Moreover, the low-level design of \Sys allows seamless integration into agent serving systems such as Continuum, where it further reduces average job latency by up to 18.1\%.

\end{abstract}

%% file: chapters/introduction.tex
\section{Introduction}

Large Language Models (LLMs) have become a cornerstone of modern data systems~\cite{DBLP:journals/pvldb/ZhouLSLCWLFZ24,DBLP:journals/pvldb/HuPK24,DBLP:journals/pvldb/ChangG25,zhang2025self,tan2025can,DBLP:journals/pvldb/HuangLZZYLZCCL25,giannakouris2025lambda,lu2025adda,DBLP:journals/pvldb/HuangCWWW26}, enabling a wide range of applications such as conversational agents~\cite{zhao2024chat2data}, retrieval-augmented generation~\cite{jiang2024chameleon}, and multi-turn interactive analytics~\cite{yan2025contextcache}.
During auto-regressive inference, LLMs maintain key–value (KV) caches to store intermediate attention states from previous tokens.
While KV caches are essential for avoiding redundant computation, they also introduce a significant memory footprint: for long contexts, multiple concurrent user sessions, and multi-turn workloads, KV caches can easily dominate GPU memory consumption.
This memory pressure poses a fundamental scalability challenge.
In practical LLM serving systems, GPU memory is often the primary bottleneck that limits the number of concurrent requests, session lengths, and throughput.
For example, for a 70B-scale LLM with a hidden size of 8,192, storing KV caches for a single 32K-token session requires over 40-GB memory in half precision, exceeding the capacity of a single A100 GPU.
Moreover, KV caches grow dynamically with sequence length and concurrent session number, making memory pressure increasingly severe in long-context and interactive workloads.
In practical LLM serving systems, GPU memory becomes the primary bottleneck that limits the number of concurrent requests, session lengths, and throughput.
To address this, a growing body of research has explored KV cache management
techniques that trade memory usage, computation cost, and model quality.

One line of work~\cite{DBLP:journals/pvldb/ChenZLZZLLLJCDYJCZLYYY26,yan2025contextcache,zhang2025pqcache,agarwal2025cache,xia2023flash,yuan2025depcache} focuses on sparse or approximate KV cache retention, aiming to reduce memory usage by keeping only a subset of KV entries deemed ``important''.
These systems leverage heuristics or learned importance metrics to discard or compress less critical KV states. Although effective in reducing memory footprint, such approximation-based techniques inevitably introduce accuracy degradation, potentially affecting generation quality and making them unsuitable for applications that require high-quality or exact model outputs.
To preserve output correctness, another mainstream direction adopts lossless KV cache management strategies. Instead of permanently discarding KV states, these systems temporarily evict KV cache blocks from GPU memory and reconstruct them on demand, either via recomputation or swapping from host memory.
This paradigm ensures bitwise-equivalent outputs while enabling more flexible memory usage.
Our work falls into this lossless KV cache management paradigm and focuses on a central design question within it: \textit{which KV cache blocks should remain resident in GPU memory when memory is constrained?}
In other words, we study the eviction policy that determines the residency of KV cache blocks, with the goal of improving serving efficiency without sacrificing output correctness.

Existing lossless KV cache management systems answer this question with different eviction policies. For example, vLLM~\cite{kwon2023efficient} and SGLang~\cite{zheng2024sglang} support prefix caching with LRU-style eviction, enabling reuse of shared prefixes such as system prompts.
Pensive~\cite{yu2025stateful} considers conversation recency and positional recomputation cost, preferentially evicting leading tokens that are cheaper to recompute.
Other systems improve cache efficiency by considering access frequency~\cite{li2025hotprefix} or layer-wise scheduling~\cite{DBLP:journals/pvldb/MaHULCLJ26}.

However, these policies do not explicitly model the future reuse likelihood of individual KV cache blocks and their marginal impact on GPU attention-kernel latency.
As a result, they cannot quantify how keeping a particular block on GPU affects subsequent attention computation, including memory access behavior, thread-block utilization, and end-to-end attention latency.
Ignoring this block-level heterogeneity leads to suboptimal eviction decisions and leaves substantial GPU performance potential untapped.

Motivated by this insight, we propose \Sys, an expected-latency-aware KV cache management system for LLM inference, designed to explicitly align KV cache residency decisions with GPU attention kernel performance. It consists of three tightly integrated components:
\textbf{(1) Multi-Segment Attention}, a novel attention kernel that decomposes attention computation into multiple non-contiguous segments, enabling fine-grained control over which KV cache blocks participate in GPU-resident computation. \textbf{(2) Computational-Aware Block Evictor}, which efficiently prioritizes KV cache eviction based on each block’s marginal contribution to expected attention latency considering both future access and position. \textbf{(3) Chunking Scheduler}, which enables prefill chunks span multiple segments and dynamically adjusts chunk sizes to balance recomputation overhead and GPU kernel efficiency under varying workloads.

We implement \Sys in vLLM and evaluate it across diverse workloads, including multi-session, multi-round conversational, and agentic scenarios with long contexts. Experimental results demonstrate that \Sys significantly outperforms state-of-the-art baselines. In particular, under realistic multi-session workloads, \Sys reduces time-to-first-token (TTFT) by up to $1.86$× and system time-per-output-token (TPOT) by up to $1.62$×, while preserving exact model outputs. 
In addition, operating at block granularity, \Sys can be seamlessly integrated into existing agent serving systems such as Continuum~\cite{li2025continuum}, which also optimize KV Cache eviction at request granularity. The evaluation demonstrates that our approach further average reduces job latency by 18.1\%, confirming its orthogonality and broad applicability.

In summary, this work makes the following contributions:

\begin{itemize}
    \item We identify and quantify the computational asymmetry of KV cache blocks in GPU attention kernels, revealing a critical gap in existing KV cache management strategies.
    \item We design \Sys, the first KV cache management system that is explicitly aware of the expected latency.
    \item We propose a novel kernel and system-level techniques that optimize both memory usage and execution efficiency.
    \item We demonstrate substantial performance gains over prior approaches through extensive experimental evaluation.
\end{itemize}

%% file: chapters/background.tex
\section{Background}
\begin{figure}[t]
    \centering
    \includegraphics[width=0.85\linewidth]{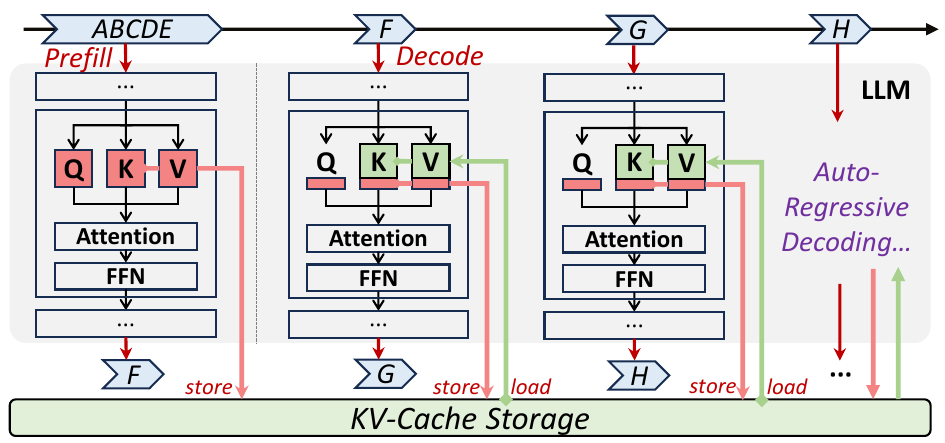}
    \caption{LLM inference with KV Cache.}
    \label{fig:intro}
\end{figure}

\subsection{LLM Inference}
Modern large language models (LLMs), such as GPT~\cite{achiam2023gpt}, Llama~\cite{dubey2024llama}, and Gemini~\cite{team2023gemini}, are built upon the Transformer architecture. Each Transformer layer primarily consists of an Attention~\cite{vaswani2017attention} module and a Feed‑Forward Network (FFN). The Attention module is pivotal to the model’s performance, with common variants including Multi‑Head Attention (MHA), Grouped‑Query Attention (GQA)~\cite{ainslie2023gqa}, and Multi-Head Latent Attention (MLA)~\cite{liu2024deepseek}. The standard attention computation can be formulated as:

\begin{equation}
\label{eq:attn}
\begin{gathered}
    Q = W_Q \cdot X, \quad K = W_K \cdot X, \quad V = W_V \cdot X \\
    A = QK^{T} / \sqrt{d_k} ,~~O = \text{softmax}(A) \cdot V
\end{gathered}
\end{equation}
During inference, an input prompt is tokenized into a sequence of tokens, each associated with its own \( Q, K, V \) vectors. Notably, the computation complexity of Equation~\ref{eq:attn} grows quadratically with the sequence length.
LLM inference follows an auto-regressive process. As illustrated in Figure \ref{fig:intro}, the model consumes the input sequence to generate the first output token. This output is then appended to the input for the next decoding step, and the process repeats until an end‑of‑sequence token is produced or the maximum output length is reached.
For a given request, the initial step, called the \textit{prefill phase}, requires computing the \( Q, K, V \) vectors for every token in the input prompt, which becomes computationally intensive for long sequences. In subsequent decoding steps, only the \( Q, K, V \) vectors of the newly generated token need to be computed, and the \( K, V \) vectors of previous tokens can be reused. To avoid redundant computation, modern inference systems~\cite{kwon2023efficient,zheng2024sglang,holmes2024deepspeed,tensorrt-llm} maintain a \textit{KV‑Cache} that stores these historical \( K, V \) states. Since the KV‑Cache can occupy substantial memory, numerous optimizations~\cite{kwon2023efficient, zhang2025pqcache,yuan2025depcache,gao2025apt} have been proposed to manage its storage and access efficiently. These subsequent steps are called \textit{decoding phase}s, and are usually memory-bounded.

\subsection{Prefix Caching}
Prefix caching is a key technique for reducing the computational overhead of the prefill phase, offering significant benefits for requests with long context . By storing the \( K, V \) vectors of previously processed tokens in KV‑Cache, system can check whether a portion of a new request’s tokens already have computed values, thereby avoiding redundant computations and accelerating prefill. 
This approach is widely adopted in modern inference engines: for instance, SGLang~\cite{zheng2024sglang} manages cache at the request level using a Radix‑tree, while vLLM~\cite{kwon2023efficient} operates at the granularity of token blocks with an LRU eviction policy. 
The cache‑key is not the individual token, but the complete token sequence from the start of the request to the current position. Consequently, existing systems typically prioritize retaining tokens from the beginning of the sequence to maximize the cache hit rate, hence the term ``prefix'' caching. Many recent works (e.g.,~\cite{wang2025kvcache,li2025hotprefix,ye-etal-2024-chunkattention}) conduct further exploration based on the prefix caching property. For example, HotPrefix~\cite{li2025hotprefix} applies hotness-aware method to schedule KV-Cache between hierarchical storage.
On the other hand, Pensieve~\cite{yu2025stateful} observes that tokens in later positions carry a higher recomputation cost and thus prioritizes caching these tokens for single‑user, multi‑turn dialogues. In contrast, \Sys holistically evaluates the trade‑off between the hit‑rate benefit of caching earlier tokens and the recomputation savings from caching later tokens, leading to a more efficient and general‑purpose cache management strategy.

\subsection{Chunked-Prefill}
Many LLM serving systems co-locate compute‑intensive prefill tasks and memory‑bound decode tasks on the same GPU instance. This can lead to head‑of‑line blocking, where long prefill requests delay decode‑stage requests and increase tail latency. To mitigate this, Sarathi‑serve~\cite{agrawal2024taming} introduced the chunked‑prefill approach, which splits a long prefill request into smaller chunks. These chunks are then interleaved with running decode requests, effectively reducing blocking and improving latency fairness. This method has been widely adopted in systems like DeepSpeed‑MII~\cite{holmes2024deepspeed} and is particularly effective for long‑sequence workloads.
Further kernel‑level optimizations have been developed to improve resource utilization under this paradigm. POD‑Attention~\cite{kamath2025pod} fuses the prefill and decode operators into a single CUDA kernel, better overlapping GPU compute and memory bandwidth. Similarly, Drift~\cite{cui2025optimizing} employs fine‑grained control over the number of Streaming Multiprocessors (SMs) allocated to different kernels and uses request partitioning to enable efficient concurrent execution of prefill and decode tasks on a single GPU.
\Sys integrates chunked‑prefill scheduling with its novel cache‑management strategy to deliver an efficient and responsive LLM inference service.

\section{Observation and Motivation}
\label{sec:motivation}
Existing LLM inference systems commonly employ KV‑cache reuse to reduce time‑to‑first‑token (TTFT). Through our analysis of real workloads and system behavior, we identify two key observations that motivate a more flexible caching strategy.

\begin{figure}[t]
    \begin{minipage}{.51\linewidth} 
        \centering
        \includegraphics[width=\linewidth]{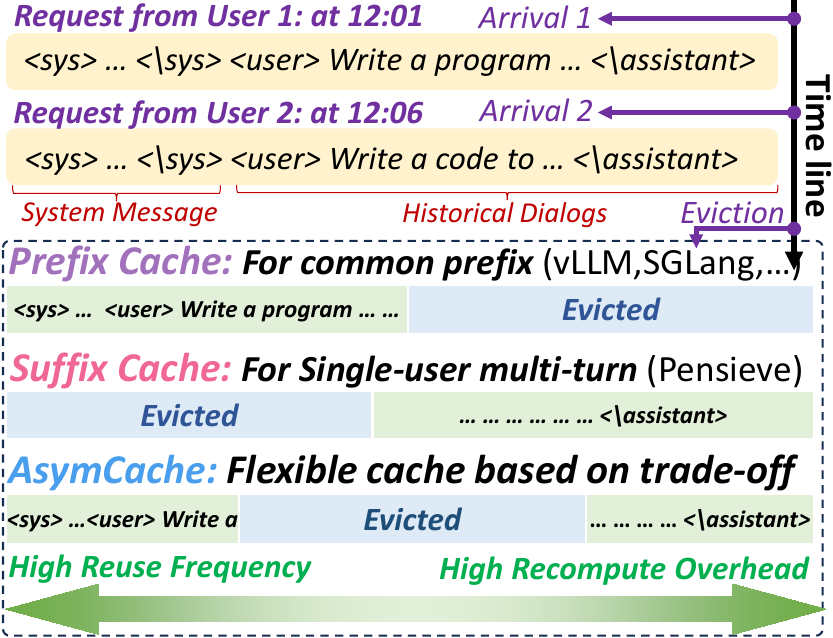} 
        \captionof{figure}{Comparison on KV cache eviction scheme.}
        \label{fig:rebuttal}
    \end{minipage}\hfill 
    \begin{minipage}{.475\linewidth} 
        \centering
        \includegraphics[width=\linewidth]{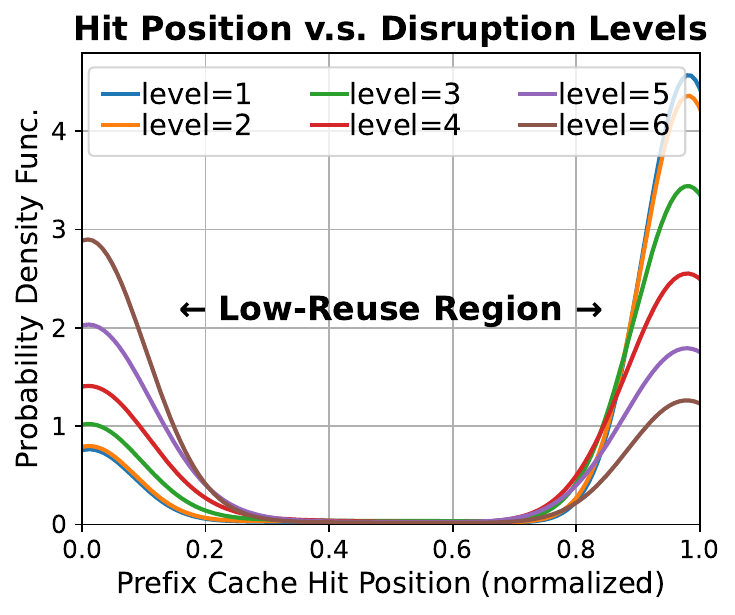}
        \captionof{figure}{PDF of normalized hit position under different disruption levels.}
        \label{fig:kde}
    \end{minipage}
\end{figure}

\paragraph{Observation 1: Both prefixes and suffixes are worth caching.}
Many systems, such as vLLM~\cite{kwon2023efficient} and SGLang~\cite{zheng2024sglang}, implement prefix caching, which reuses KV blocks across different requests that share common initial text (e.g., system messages or dialogue templates). This cross‑request sharing yields high hit rates and eliminates recomputation for shared prefixes. On the other hand, Pensieve~\cite{yu2025stateful} targets single‑user multi‑turn dialogues where cross‑request prefix sharing does not occur. In that scenario, caching suffix tokens is beneficial because the quadratic complexity of self‑attention makes recomputing later positions increasingly expensive. Pensieve shows that suffix‑only caching can be optimal in its specific setting. However, in general‑purpose serving environments, both patterns coexist: prefixes are shared across independent requests, while suffixes are reused within long‑running sessions. Therefore, an effective cache management policy is supposed to retain both: prefix blocks that serve many requests, and suffix blocks that are costly to recompute, rather than committing to a single strategy.

\paragraph{Observation 2: Cache hits tend to cluster near the two ends of the sequence.} 
Analyzing multi‑turn traces (Figure~\ref{fig:kde}), we observe that a cache hit often falls into one of two patterns.
One pattern is a short prefix shared across many requests, such as system messages or dialogue templates. The other pattern is an almost complete sequence, typically when a long conversation history is reused in subsequent turns.
The middle portion of the sequence accounts for comparatively fewer cache hits. This bimodal distribution suggests that evicting middle blocks while retaining tail blocks sacrifices only limited hit rate but achieves substantially lower recomputation cost.
For this particular workload, an ideal eviction policy could be to drop the middle and preserve both ends. However, real world workloads can exhibit more complex reuse patterns, sometimes with multiple distinct peaks or irregular distributions.
A cache management policy that sticks to a single fixed heuristic is unlikely to be optimal across all scenarios. Instead, an ideal eviction algorithm should be able to automatically identify the underlying cache reuse pattern of the current workload and adapt its strategy accordingly, whether that means preserving prefixes, suffixes, both ends, or other complicated structures.

These observations motivate a multi-segment caching strategy that prioritizes retaining tokens near the two ends of the sequence while selectively evicting less critical middle segments. 
This approach is shown to achieve significant performance gains with only marginal degradation in the overall hit rate.

Figure \ref{fig:overview} illustrates the overall architecture of \Sys. The system realizes this multi‑segment caching strategy through three core components: a Multi‑Segment Attention (MSA) kernel (\S\ref{subsec:msa}), a carefully designed block‑management algorithm (\S\ref{subsec:evictor}–\ref{subsec:freq}), and Chunking Scheduler (\S\ref{subsec:runtime}). 
When user submits a query, the query text may first be concatenated with relevant context (e.g., historical dialogue) and sent to the API Server (step $\textcircled{1}$). Upon being scheduled (step $\textcircled{2}$), \Sys checks its cache‑hit status and allocates cache blocks, a process that may trigger cache eviction. Due to the multi‑segment design, a request can have multiple discontiguous cache‑hit segments (step $\textcircled{3}$). The request is then batched with other running requests, and the Chunking Scheduler assigns a chunk size to each prefill request in the batch (step $\textcircled{4}$). Because of the non‑contiguous KV‑cache context, some prefill chunks may span already‑computed segments, as shown for Prefill Request 1 in Figure \ref{fig:overview} (step $\textcircled{5}$). Finally, based on this scheduling plan, the batch is dispatched to the inference engine for execution (step $\textcircled{6}$), and promote to next schedule step (step $\textcircled{7}$).

\begin{figure}[t]
    \centering
    \includegraphics[width=1.0\linewidth]{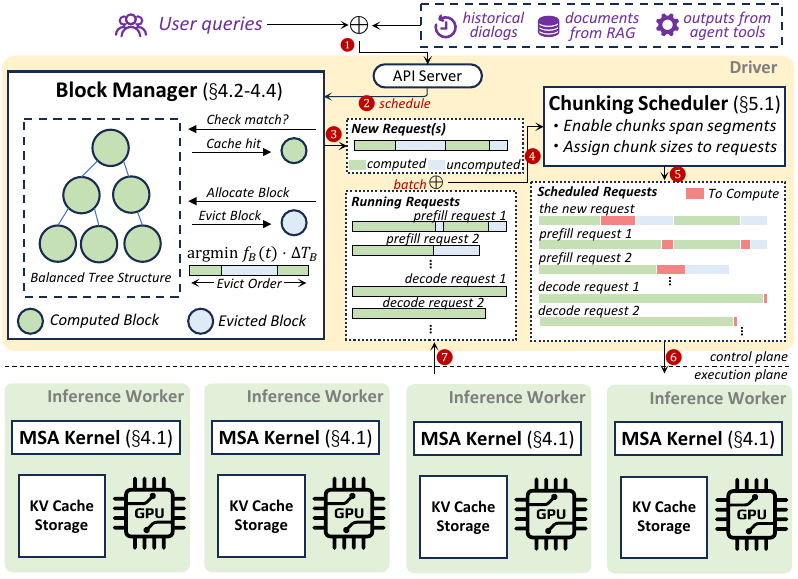}
    \caption{The overview of \Sys.}
    \label{fig:overview}
\end{figure}

%% file: chapters/design.tex
\section{Asymmetric Cache Block Manager}

\subsection{Multi-Segment Attention}
\label{subsec:msa}
Supporting non-contiguous cache segments (e.g., both a prefix and a suffix) requires the attention kernel to handle disjoint KV regions. 
This motivates our design of Multi-Segment Attention (MSA).
In contrast to prefix caching strategy which only caches prefix and prioritizes the blocks with longer prefix to be evicted, MSA allows multiple cache segments in latter position for less computational overhead. 
It can be considered as a combination of multiple individual self-attention. For a two-segment instance, letting $h_1$ and $h_2$ denote the hidden states of these two segments respectively, the compute procedure can be formulated as:

\begin{equation}
\begin{gathered}
    H=Concat(h_1,h_2) \\
    Q=W_Q\cdot H, K=W_K\cdot H, V=W_V\cdot H \\
    Q_1=Q[:len(h_1)],K_1=K[:len(h_1)],Q_2=Q[len(h_1):]\\
    O_1=Attn(Q_1,K_1,V), O_2=Attn(Q_2,K,V)\\
    O=Concat(O_1,O_2)
\end{gathered}
\end{equation}

One straightforward way to compute MSA is utilizing existing attention kernel implementations and dispatching each segment to separate attention kernel calls. 
Several LLM serving systems prefer this method for computing attention in hybrid batches~\cite{yu2022orca,kwon2023efficient}.
However, separate kernel calls introduce additional kernel launch and host synchronization overheads, leading to sub-optimal performance.
Some alternative solutions, such as multi-streams, CUDA GreenContext~\cite{chen2026towards}, and kernel fusion, also face performance challenges like resource contention and stragglers.

To achieve maximum GPU resource utilization, we implement a single GPU kernel that efficiently computes multiple segments' attention simultaneously.
In particular, we fuse computation along the Cooperative Thread Array (CTA) dimension and each request segment is dispatched to a group of CTAs for better parallelism.
The GPU hardware will automatically assign CTAs to available Streaming Multiprocessors (SMs).
Note that, only the attention module requires the token dependency information, so the hidden states of two segments can directly be concatenated when computing MLP and LayerNorm, while the only change is passing their split indices to the Multi-Segment Attention kernel, maintaining the number of kernel calls and keeping the modification lightweight.

Figure \ref{fig:msa-layout} shows a batch consisting of two double-segment cached prefill requests. The MSA kernel dispatches different requests to their respective CTA groups for independent execution. Within each thread, the equivalent \texttt{seq\_len} of the token block it is responsible for is obtained via a precomputed array. Ultimately, parallel computation for multiple requests is achieved within a single kernel call, which can be extended to any number of non‑contiguous cache segments.

\begin{figure}[t]
    \centering
    \includegraphics[width=1.0\linewidth]{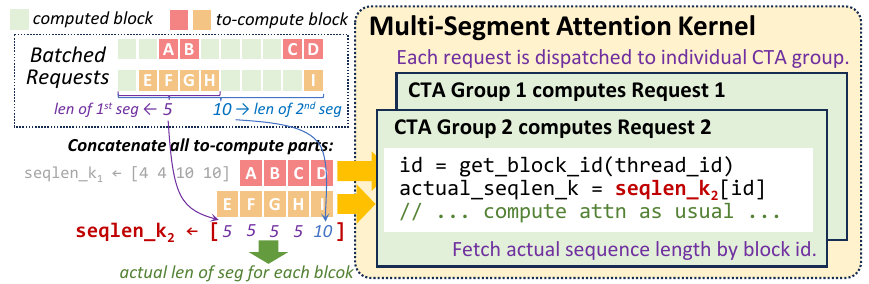}
    \caption{Schematic diagram of the Multi-Segment Attention Kernel operation. (An example for 2 prefill requests.)}
    \label{fig:msa-layout}
\end{figure}

\subsection{Computational-Aware Block Evictor}
\label{subsec:evictor}
Caching the blocks with longer prefix helps reduce the prefilling latency. However, since different requests may share a common prefix, caching these blocks in the latter position leads to hit rate loss.
We introduce the trade-off mechanism in \Sys, which takes both cache hit rate and asymmetric re-compute overhead into consideration, making the eviction policy expected-latency-aware that select the block with the least expected computational latency to evict. 
Specifically, \Sys incorporates the \textit{frequency} and the computational \textit{cost} values for each block.
The \textit{frequency} value of a block is its historical access frequency with exponential weight decay, as originally used in most Least Frequently Used (LFU) evict policy, while the block's \textit{cost} measures the additional prefilling time if this block is evicted.

Whenever scheduling a new request that exceeds the GPU memory constrains, the Block Evictor should select certain free blocks to evict and assign them to the coming request. Least Recently Used (LRU) or Least Frequently Used (LFU) policy, only considering the block accessing time, is applied in existing systems such as~\cite{kwon2023efficient,zheng2024sglang,qin2024mooncake}, choosing blocks with the least scores to evict. 
In \Sys, we extend LFU to incorporate the cost value of blocks.

We suppose $f_B(t)$ is the frequency value of KV block $B$ at time $t$, and $\Delta T_B$ is its cost. When the eviction operation is launched by scheduler, the Block Evictor is supposed to choose the cache block with the lowest expected recomputation cost for eviction. Based on the assumption that the probability of a cache block being used in the future is proportional to its historical usage frequency $f_B(t)$, this least-expected-latency selection criterion can be formulated as:
\begin{equation}
\label{eq:exp}
\mbox{block\_id} \gets \argmin_{B}~E(B,t)=\argmin_{B} f_B(t)~\cdot \Delta T_B,
\end{equation}

In Equation \ref{eq:exp},
\noindent where the frequency value $f_B(t)$ is generally a decreasing function with respect to the last access time $t$. If the recovery cost of all cache blocks is a uniform constant, our algorithm degrades to the conventional LRU strategy. Some eviction policies may also take computational information into account. For example, Pensieve~\cite{yu2025stateful} uses the historical usage frequency as the frequency value, similar to the conventional LFU algorithm, while a recent work~\cite{wang2025kvcache} introduces a new approach that combines the last access time and the lifespan to estimate the priority of eviction. 
However, the definitions of the frequency value $f_B(t)$ in these algorithms are straight-forward, causing the 
linear-time complexity when calculating the priority of each cached block and selecting the block to be evicted. 
This is highly inefficient compared to the $O(1)$ complexity of the LRU algorithm. Our system adopts a piecewise exponential function scheme to define $f_B(t)$, which ensures that both the recent access time and lifespan of each cached block are taken into account while also achieving fast runtime performance. 
We will discuss this scheme in detail in \S\ref{subsec:freq}.

The recomputational cost term $\Delta T_B$ in Equation \ref{eq:exp} is what distinguishes our algorithm from traditional LRU
approaches. According to \textit{Observation 1} in \S\ref{sec:motivation}, different from treating each cache block equivalently, the $\Delta T_B$ is positively correlated to the block's predecessor length. 
On the one hand, for cache blocks with identical frequency values, those associated with lower recomputation costs are prioritized for eviction. This design inherently favors the retention of blocks located at later positions in the request sequence, which typically incur greater computational overhead. 
On the other hand, cache blocks with higher access frequencies are also more likely to be retained, as these often constitute shared prefixes common across multiple requests. Consequently, a single request sequence may end up hitting multiple (typically two) non-contiguous segments of the KV context. 
This is precisely where our proposed Multi-Segment Attention (MSA) mechanism is effectively leveraged.

\subsection{Cost Model Estimation}
In practical operation, the system must respond to cache eviction requests and make eviction decisions within extremely short time frames. Consequently, the estimation of recomputation cost cannot be overly complex. \Sys employs a linear model to approximate this cost effectively.
Although prefill requests may encounter non-contiguous KV contexts, resulting in the segmentation of query tokens into discontinuous segments, the non-attention components of the model (such as MLP and LayerNorm) process each query token in parallel and independently, without modeling inter-token relationships. Consequently, the inference time of these modules
scales linearly with the total number of query tokens in the request. 
For the attention module, the multiplicative operations between query and key matrices introduce a quadratic term into its computational complexity. 
In summary, taking the two-segment cache scenario as an example, the inference latency can be linearly modeled as follows:
\begin{multline}
T(l_1,q_1,l_2,q_2) = k_1\cdot l_1 + k_2\cdot q_1 + k_3\cdot l_2 + k_4\cdot q_2 \\
 + k_5\cdot q_1 (l_1+q_1) + k_6\cdot q_2 (l_1+q_1+l_2+q_2)+ \beta
\end{multline}
When the eviction policy selects the first cache block in the second segment of the KV context (denoted by the red block in Figure \ref{fig:delta-t}), its recomputation cost is given by:
\begin{equation}
\begin{aligned}
\label{eq:deltat}
\Delta T_B&=T(l_1,q_1+1,l_2-1,q_2) - T(l_1,q_1,l_2,q_2)\\
& = k_5\cdot(l_1 + 2q_1) +(k_2-k_3+k_5)
\end{aligned}
\end{equation}
However, maintaining the term $(l_1 + 2q_1)$ in Equation \ref{eq:deltat} would require the introduction of complex data structures and involve update or query operations with super-constant time complexity, becoming unaffordable in online serving. 
Therefore, the following approximated model is applied instead:
\begin{multline}
T(l_1,q_1,l_2,q_2) = k_1\cdot l_1 + k_2\cdot q_1 + k_3\cdot l_2 + k_4\cdot q_2 \\
 + k_5\cdot (l_1+q_1)^2 + k_6\cdot q_2(l_1+q_1+l_2+q_2)+ \beta
\end{multline}
Then, the recomputation cost becomes:
\begin{equation}
\label{eq:deltat2}
\Delta T_B=2k_5\cdot(l_1 + q_1) +(k_2-k_3+k_5)
\end{equation}
The term $(l_1 + q_1)$ corresponds to the immutable positional index of the cache block within the request sequence (i.e., the number of preceding blocks), which can be retrieved in constant time. The approximated models were fitted using approximately 1.1K real-world performance profiling instances under various configurations. All entries achieve notably high correlation coefficients $R^2>0.999$, validating the soundness of our design.

\begin{figure}[t]
    \centering
    \includegraphics[width=0.85\linewidth]{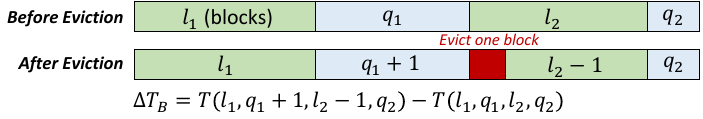}
    \caption{Estimation of $\Delta T_B$.}
    \label{fig:delta-t}
\end{figure}

\subsection{Frequency Value Definition}
\label{subsec:freq}
In Equation \ref{eq:exp}, the frequency term \(f_B(t)\) assigns a weight to each cache block based on historical access information. Cache blocks with smaller weights are prioritized for eviction by the cache eviction mechanism. Consequently, \(f_B(t)\) is generally 
correlated with the historical access frequency and 
the last access time. A straightforward approach is to directly set it as the historical access frequency, aligning with prior work such as Pensieve~\cite{yu2025stateful}. However, this leads to suboptimal time efficiency.

\textit{Requirement 1: Order-Preserving Rule.} 
The cache eviction algorithm needs to calculate the weight for each cache block according to Equation \ref{eq:exp} and select the block with the smallest weight for the eviction operation (\textit{evict}). The system must maintain a data structure that stores this weight for each cached block. Furthermore, this data structure must efficiently support two other operations: inserting a new block (\textit{insert}), removing a block when it is accessed (\textit{remove}), and identifying the block with the minimum weight for eviction. 
A balanced tree can perform all three operations in logarithmic time, avoiding the need for linear time complexity. However, this choice imposes an order-preserving rule on the weight values. 
For any cache blocks $B_1$ and $B_2$ and any time points $t_1$ and $t_2$, the order relation between their weights should remain consistent:
\begin{equation}
\label{eq:order}
\begin{aligned}
& \left(f_{B_1}(t_1)\cdot \Delta T_{B_1} - f_{B_2}(t_1)\cdot \Delta T_{B_2}\right)\cdot \\
& ~~~~\left(f_{B_1}(t_2)\cdot \Delta T_{B_1} - f_{B_2}(t_2)\cdot \Delta T_{B_2}\right)\geq 0
\end{aligned}
\end{equation}
The equality holds if and only if both terms in the product are zero. If this order-preserving rule is violated, the partial ordering among elements in the balanced tree may change over time, rendering the logarithmic-time algorithm inoperative. It can be verified that using the historical access frequency as the frequency term violates this property, which would degrade the time complexity of the eviction algorithm to linear complexity. Given that the frequency value is monotone and non‑negative, it can be proved that only exponential functions satisfy the above property. (See our artifact for proof).

\begin{figure}[t]
    \centering
    \includegraphics[width=1.0\linewidth]{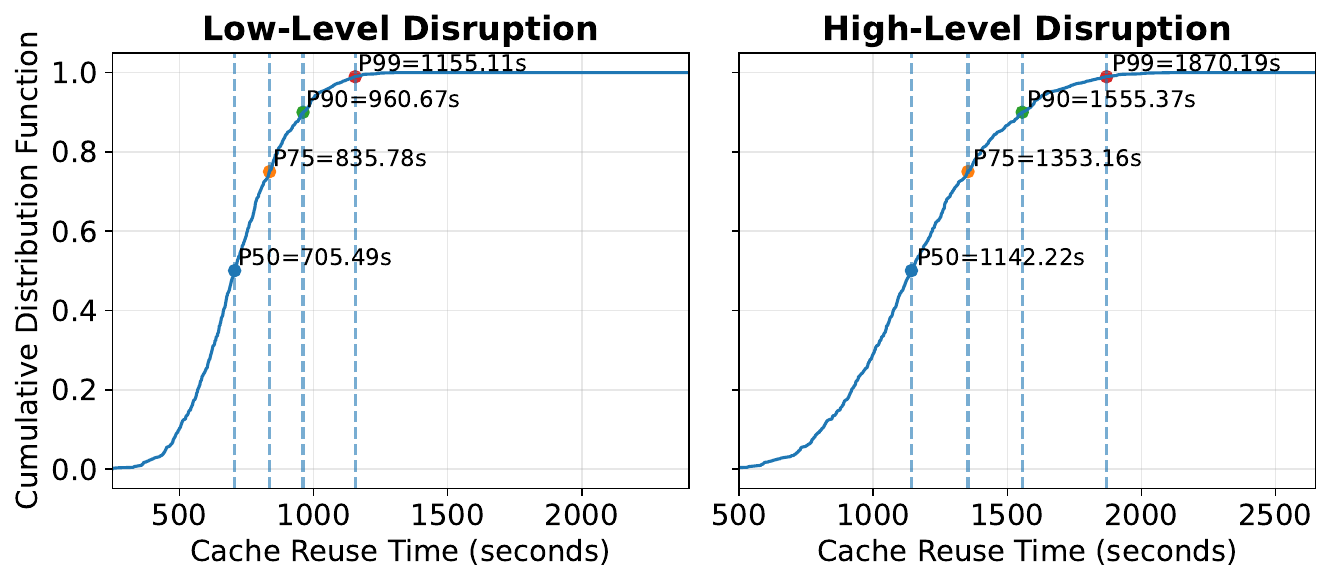}
    \caption{KV-Cache reusing time distribution of LooGLE~\cite{li2024loogle} dataset in different disruption level.}
    \label{fig:cdf}
\end{figure}

\textit{Requirement 2: Alignment with the Reuse Pattern.}
As illustrated in Figure \ref{fig:cdf}, most requests still have a high probability of being reused in a certain period of time after their initial access. In a multi-turn dialogue, for example, a user may continue the conversation with a follow-up request within a short period after receiving the response to the previous question. It is only after a certain time threshold that the reuse frequency of a cache block drops significantly. This threshold can be regarded as the cache block's \textit{lifespan}, a concept also noted in existent works~\cite{wang2025kvcache,li2025continuum}. The frequency value function should align with such a pattern.

\begin{figure}[t]
    \begin{minipage}{.51\linewidth} 
        \centering
        \includegraphics[width=\linewidth]{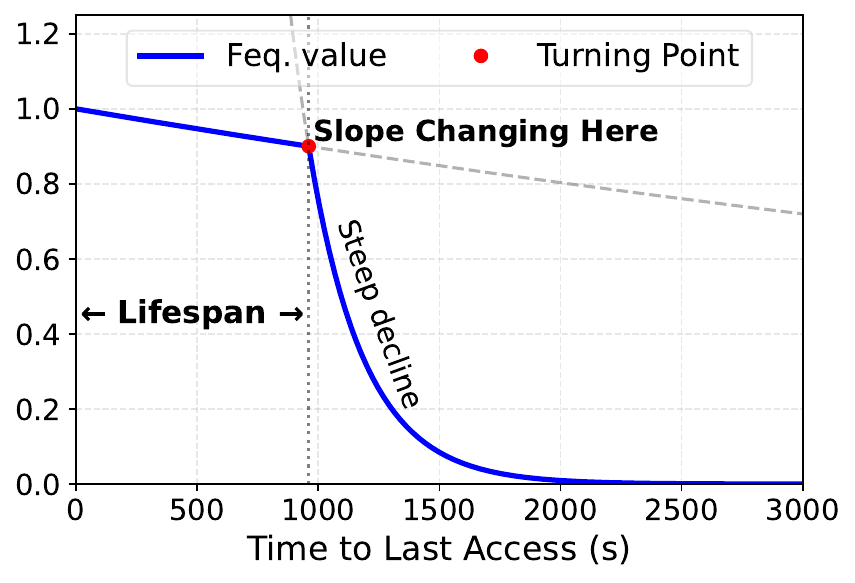} 
        \captionof{figure}{Image of piecewise exponential function.} 
        \label{fig:left}
    \end{minipage}\hfill 
    \begin{minipage}{.475\linewidth} 
        \centering
        \includegraphics[width=\linewidth]{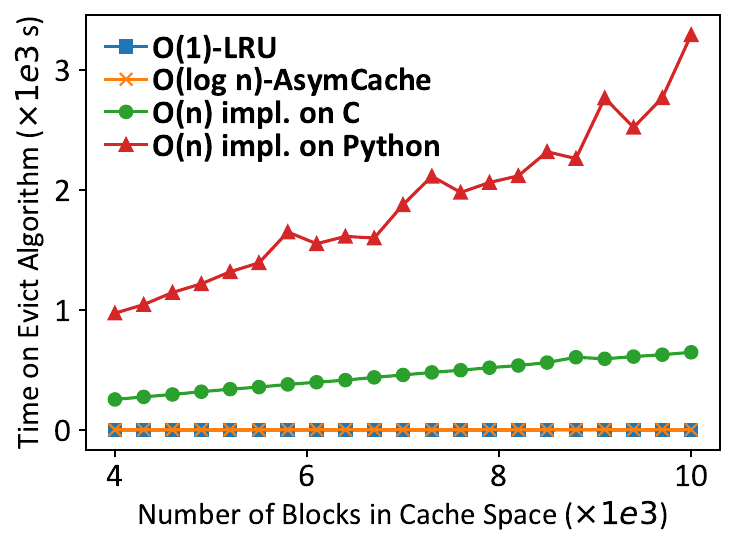}
        \captionof{figure}{Performance of eviction algorithms}
        \label{fig:right}
    \end{minipage}
\end{figure}

\begin{algorithm}[t]
\small
\begin{algorithmic}[1]
\State Initialize two balanced trees $bt_1,~bt_2$.
\LineComment{Called when the ref-count of Block $B$ becomes zero.}
\Function{Add}{$B$}
    \State $B.w_1\gets\exp\left(-\tau_B(t)/\alpha\right)\cdot\Delta T_B$
    \State $B.w_2\gets\exp\left(-(\tau_B(t)-\tau_0)/\beta\right)\cdot\Delta T_B$
    \State $bt_i.\textsc{insert}(\{B.\texttt{id}, B.w_i\}),~\text{for all }i\in\{1,2\}$
\EndFunction
\LineComment{Called when Block Manager requests a new block.}
\Function{Evict}{}
    \State $B_1,~B_2 \gets bt_1.\textsc{Get\_min()},~bt_2.\textsc{Get\_min()}$
    \State $i \gets \argmin \{B_1.w_1, \lambda \times B_2.w_2\}$
    \State $\textsc{Remove}(B_i)$ and then \Return $B_i$
\EndFunction
\LineComment{Called when Block $B$ is hit or from the \textsc{Evict} function.}
\Function{Remove}{$B$}
    \State $bt_i.\textsc{delete}(\{B.\texttt{id}, B.w_i\}),~\text{for all }~i\in\{1,2\}$
\EndFunction

\end{algorithmic}
\caption{Computational-aware block evictor.}
\label{algo:evictor}
\end{algorithm}

\textit{Solution: Piecewise Exponential Function.}
\Sys employs a piecewise exponential function as the frequency term to better model the reuse pattern. The first segment fits the period of high reuse probability within
lifespan, while the second segment captures the rapid decline thereafter. This can be formulated as:
\begin{equation}
\label{eq:pw-score}
f_{B}(t) := \min\left(\exp\left(-\tau_B(t)/{\alpha}\right),~\exp\left(-(\tau_B(t)-\tau_0)/{\beta}\right)\right)
\end{equation}
Here, the hyper-parameters $\alpha$ and $\beta$ control the bases (and thus the decay rates) of the respective exponential functions. The hyperparameter $\tau_0$ controls the horizontal shift of the second segment. 
Figure \ref{fig:left} shows the sketch of $f_B(t)$ defined by Equation \ref{eq:pw-score}. 
The frequency value keeps high during the lifespan, while exhibits a steep decline beyond the period.
These parameters can be uniquely determined by setting the coordinate of the \textit{turning point} (e.g., at the P99 point of the cumulative distribution function) and specifying the \textit{slope changing ratio} (e.g., 40 in our evaluation) there. Experiments in \S\ref{subsec:hyper-exp} show that the system's overall performance is not highly sensitive to these hyper-parameters over a broad range.

Although the piecewise function as a whole violates the order-preserving rule specified in Equation \ref{eq:order}, each individual segment is an exponential function and thus satisfies this property. Consequently, in the actual implementation of the system, two separate balanced trees can be used to maintain the values for each segment. The algorithm \ref{algo:evictor} shows the overall workflow for cache management, with all operations within logarithmic time complexity, preserving the overall efficiency of the algorithm.
In Line 8 of Algorithm~\ref{algo:evictor}, we introduce a coefficient $\lambda$ that is initially set to 1, making the eviction criterion identical to the original frequency value definition. By adjusting $\lambda$ online, the effective turning point of the piecewise exponential function can be shifted, enabling \Sys to adapt to workloads whose lifespan characteristics vary over time. Further details are presented in Section~\ref{subsec:runtime}.

\subsection{Time Complexity}
\label{subsec:time_complex}
The time complexity of the eviction algorithm is dominated by the operation defined in Equation~\ref{eq:exp}. Both the frequency values and the recomputation costs of cache blocks can be retrieved in $O(1)$ time. Moreover, our carefully designed weight formulation enables balanced-tree optimization, ultimately bounding the overall time complexity of the cache management module to $O(\log n)$, exhibiting clear performance advantages compared with existing $O(n)$ cache eviction algorithms. 
Such a linear overhead becomes prohibitive as cache size grows. As depicted in Figure~\ref{fig:right}, when serving \(\sim\)2.4~K requests with 8~K cache blocks, an \(O(n)\) algorithm (implemented in C) consumed \(\sim\)500~seconds of control‑plane time, (i.e., roughly 200~ms per request), or nearly 10\% of the prefill latency. Such overhead may be undesirable in an online serving system.

\Sys's logarithmic‑time eviction algorithm enabled by the carefully designed frequency function that satisfies the order preserving rule (\S\ref{subsec:freq}). Under this rule, the relative ordering of weights between any two blocks remains invariant over time, allowing us to maintain all cached blocks in a balanced tree keyed by their current weight. Insertion, deletion, and finding the minimum weight block each take \(O(\log n)\) time. In practice, this reduces the total eviction overhead to merely \(\sim\)5~seconds for the same 2.4~K requests, i.e., approximately 2~ms per request, or 0.1\% of prefill time, matching the efficiency of classical LRU. More importantly, the logarithmic complexity scales gracefully to larger cache configurations (e.g., \(>\)100~K blocks, which may appear when offloading KV caches to CPU memory), whereas linear overhead would become completely untenable. We emphasize that this algorithmic choice keeps control‑plane overhead negligible, allowing the eviction policy’s computational‑awareness to deliver its latency benefits without becoming a bottleneck itself.

%% file: chapters/scheduler-new.tex
\section{\Sys Runtime}
This section describes the runtime of \Sys, emphasizing how its components collectively enable the expected-latency-aware KV cache management. We present the runtime mechanisms that support non‑contiguous context with adaptive chunking and enable adjusting lifespan value online (\S\ref{subsec:runtime}), then discuss how \Sys adapts to representative long‑context workloads (\S\ref{subsec:application}). Finally, we detail the implementation of the inference engine and its integration with the rest of the system (\S\ref{subsec:impl}).

\subsection{Runtime Design}
\label{subsec:runtime}
\paragraph{Adaptive Chunking.} 
\Sys adopts a chunked‑prefill~\cite{agrawal2024taming,holmes2024deepspeed} scheduler that splits long prefill requests into smaller chunks, preventing head‑of‑line blocking for concurrent decode requests. Unlike conventional prefix caching where each chunk corresponds to a contiguous prefix of the KV context, our multi‑segment caching policy (\S\ref{subsec:msa}) may cause a single prefill chunk to start in a region that needs computation, pass through a cached segment, then enter a second region that also requires computation (e.g., the prefill request 1 in Figure~\ref{fig:overview}), and in general include any number of alternating compute and cached segments. The MSA kernel natively supports such non‑contiguous attention computation, allowing the scheduler to treat these hybrid chunks uniformly.
When decode requests are sufficiently numerous (i.e., their count exceeds a configurable threshold), the chunking scheduler deliberately reduces the chunk size for ongoing prefill requests. Because prefill is compute‑intensive, the total latency of processing all chunks of a prefill request remains largely unchanged after chunk size reduction. However, each smaller chunk executes faster, thereby lowering their time‑per‑output‑token. To avoid under-utilization of GPU compute resources, the scheduler enforces a lower bound on chunk size. The effectiveness of this adaptive strategy is most pronounced under high‑load conditions (see \S\ref{subsec:e2e}).

\paragraph{Online Lifespan Setting.}

The piecewise exponential frequency function described in \S\ref{subsec:freq} requires a fixed turning point (the lifespan) and a slope change ratio, both determined before system deployment. While these parameters work well for workloads with stable reuse patterns, they become suboptimal when the average lifespan of cache blocks changes significantly over time.
For example, during a sudden burst of high load, requests may be temporarily queued, introducing extra waiting time before they are processed.
To address this limitation, \Sys introduces an online adjustable coefficient \(\lambda\) in the eviction criterion (Line 8 of Algorithm~\ref{algo:evictor}). The original eviction weight is \(f_B(t) \cdot \Delta T_B\); we modify it to \(\lambda \cdot f_B(t) \cdot \Delta T_B\) with \(\lambda\) initially set to 1. Changing \(\lambda\) online equivalently rescales the frequency term, which shifts the effective turning point of the piecewise exponential function. In practice, \Sys can periodically collect the average lifespan $\tau$ of cache blocks from a sliding window and updates \(\lambda\) according to the following rule, adjusting the turning point to the detected lifespan:
\begin{equation}
\lambda_{new} \gets \exp\left( (\tau - \tau_0)/{\beta} - {\tau}/{\alpha} \right).
\end{equation}
The key advantage of this design is its zero overhead. The coefficient \(\lambda\) is a scalar multiplier applied during weight computation. Modifying it does not require any change to the balanced tree structure or the \(O(\log n)\) eviction operations. The new \(\lambda\) takes effect immediately in the next eviction decision, seamlessly adapting the policy without additional control plane cost. This lightweight adaptation mechanism enhances \Sys's robustness against workload dynamics, providing flexibility beyond fixed parameter settings.

\subsection{Applications}
\Sys exploits the inherent asymmetry in recomputation costs across sequence positions. While our design is general purpose and applies to any LLM serving workload, the benefits are most pronounced when context lengths are long. 
Two representative scenarios that exhibit these characteristics are multi‑turn conversation and agentic workload (Figure~\ref{fig:workload-cmp}), where long histories and frequent context reuse are common.

\paragraph{Multi-turn Conversation.}

In multi‑turn dialogue workloads, each new request appends the entire conversation history to the input. As the number of turns grows, the context length can become extremely long, while the prefix (e.g., system prompts or dialogue templates) is often shared across requests from different users. This yields two complementary caching opportunities. First, caching the shared prefix across requests reduces prefill latency for many sessions simultaneously. Second, due to the long context, retaining suffix blocks near the end of the sequence carries high recomputation savings due to the quadratic cost of attention. However, the inter‑turn arrival pattern is irregular: whether a subsequent turn will occur is uncertain, and the time between consecutive turns of the same conversation follows a random distribution and relatively longer. Therefore, the effective lifespan of a cache block tends to be relatively long, and usually with higher dispersion. \Sys accommodates this by setting the turning point of the piecewise exponential frequency function to a conservative P99 value of the observed lifespan distribution (as stated in Figure~\ref{fig:cdf} and \S\ref{subsec:freq}). This choice ensures that blocks with moderately old but still reusable history are not evicted prematurely, while blocks that have truly become stale are quickly de-prioritized.

\label{subsec:application}
\begin{figure}[t]
    \centering
    \includegraphics[width=0.98\linewidth]{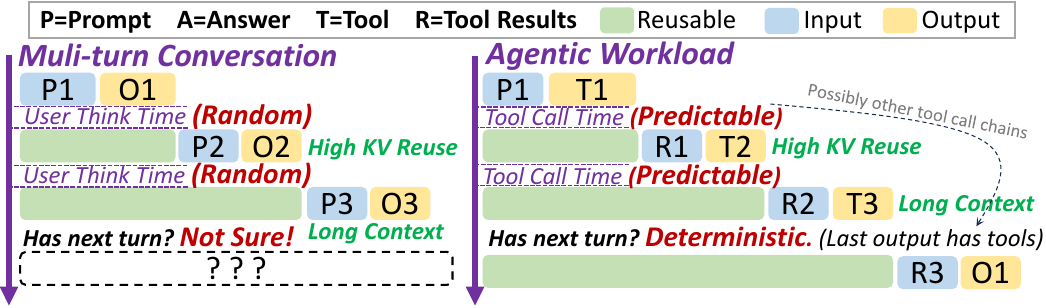}
    \caption{Comparison of two types of workloads.}
    \label{fig:workload-cmp}
\end{figure}

\paragraph{Agentic Workload.}
Agentic workloads differ from multi‑turn conversations in several key aspects. While both involve long contexts and exhibit high cache reuse, agent sessions are driven by tool calling workflows. The model iteratively produces tool invocations, receives external results, and continues generation. This leads to two distinctive characteristics. First, the intervals between consecutive turns are typically shorter and more predictable than in human‑driven dialogues. Second, whether a subsequent turn will occur is often deterministic based on the presence of a tool call in the model output. Consequently, the effective lifespan of cache blocks in agent workloads is generally shorter and less variable.
\Sys adapts to these characteristics in two ways. 
First, the piecewise exponential frequency function uses a turning point with shorter lifespan to reflect the tighter reuse window. 
Second, depending on whether the model output contains a tool call, a correction factor will be applied to the frequency value. This strongly discourages eviction of blocks that are likely to be reused immediately after the tool returns, preserving the KV-Cache of the current turn.
Importantly, \Sys operates at the granularity of individual cache blocks and focuses on the eviction order within a request. 
This is orthogonal to existing agent serving systems that optimize at the scheduling layer (e.g., Autellix~\cite{luo2025autellix}, Tokencake~\cite{bian2025tokencake}, ThunderAgent~\cite{kang2026thunderagent}) or manage KV caches at the request level (e.g., InferCept~\cite{abhyankar2024infercept}, Continuum~\cite{li2025continuum}, KVFlow~\cite{pan2026kvflow}). Our block‑level, position‑aware eviction can be seamlessly integrated into such systems. In our evaluation (\S\ref{subsec:agent-exp}), we select Continuum, which also optimizes KV‑cache management but specifically for agent scenarios, and implement \Sys on top of it. The results show that integrating our approach yields additional performance gains beyond what Continuum alone achieves.

\subsection{Implementation}
\label{subsec:impl}
We implement \Sys by extending vLLM~\cite{kwon2023efficient}, with approximately 6K lines of code (LoC) modified in Python and 2K LoC in C++/CUDA. The proposed Multi-Segment Attention (MSA) mechanism is built upon CUDA~\cite{nvidia_cuda_2024} and the CUTLASS~\cite{nvidia_cutlass_2024} library, extending FlashAttention and PagedAttention to support non-contiguous KV Cache contexts.
To maximize the performance of the cache eviction algorithm, we implement it in C++. 
In contrast, the logic for matching requests to multi-segment caches and the generation of metadata for the attention kernels remain on the Python side. 
This is because these components require extensive interaction with other system modules, such as the scheduler and the Block Manager, thus coupled with vLLM's original architecture.

%% file: chapters/evaluation.tex
\begin{table}[t]
\centering
\small
\begin{tabular}{c|cccc}
    \hline
    \toprule
    Model & \#Layer & Hidden & \#GPU & Cache Space\\
    \midrule
    Llama 3.1‑8B & 32 & 4096 & 1 & 487,744 tokens \\
    Llama 3.1‑70B & 80 & 8192 & 4 & 505,152 tokens \\
    \bottomrule
\end{tabular}
\caption{The metadata and cache space for each model.}
\label{tab:models}
\end{table}

\begin{figure*}[pt]
    \centering
    \includegraphics[width=1.0\linewidth]{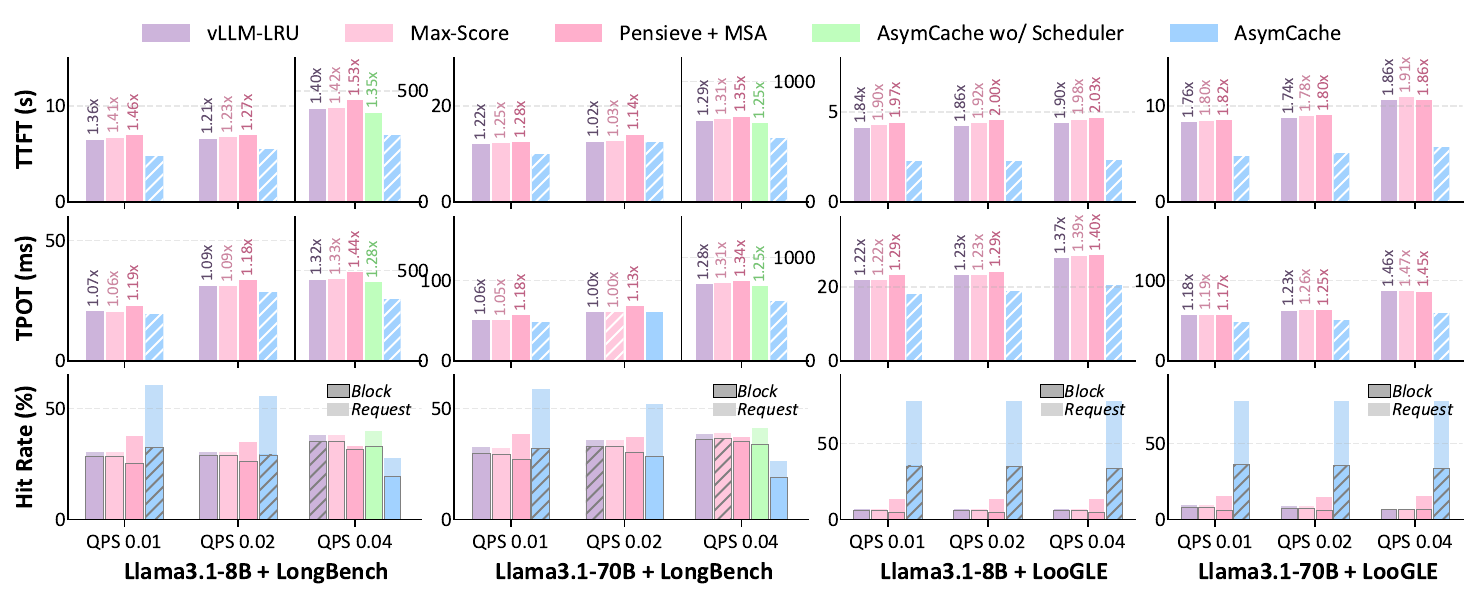}
    \caption{End-to-end results on Low-Dispersion Workloads. The best performance is denoted by the bars with diagonal slashes.}
    \label{fig:e2e-f5}
\end{figure*}
\section{Evaluation}

\subsection{Experimental Setup}
\paragraph{System Environment.}
We evaluate \Sys on a high performance server equipped with an AMD EPYC 9K84 96‑core CPU and 4× NVIDIA H20 GPUs, each with 96 GB of device memory. The GPUs are connected via NVLink to enable high‑bandwidth peer‑to‑peer communication. The software uses CUDA 12.8 and runs in a container without explicit CPU or memory constraints.

\paragraph{Models.}
We employ two open‑source models of different scales: Llama 3.1‑8B-Instruct and Llama 3.1‑70B-Instruct~\cite{dubey2024llama}. Both models adopt Grouped‑Query Attention (GQA)~\cite{ainslie2023gqa}, which reduces the number of KV heads and thereby significantly decreases the memory footprint of the KV cache. This allows \Sys to allocate more cache space for storing historical context. For the larger 70B model, we apply Tensor Parallelism~\cite{shoeybi2019megatron} to distribute it across multiple GPUs. Table~\ref{tab:models} lists the metadata of each model and the corresponding cache space sizes used in the experiments. 

\paragraph{Dataset.}
We evaluate \Sys on two long‑context datasets: LongBench~\cite{bai2024longbench} and LooGLE~\cite{li2024loogle}. Both of them are benchmark for long‑context understanding, containing multiple question answering (QA) pairs. We construct multi-turn conversation requests based on their QA pairs, similar to prior works~\cite{li2025loopserve}. To eliminate randomness from token sampling, we first generate the model output and then rewrite each output token in every decoding step, ensuring the output length remains consistent across all experimental runs. In our experiments, each dataset is limited to 300 requests, as the system behavior observed on a representative subset closely matches that on the full dataset, while significantly reducing the required time. The average input/output lengths are approximately 34.8K/2.6K for Longbench and 24.4K/0.7K for LooGLE.

\paragraph{Workloads.}
The arrival time of the first turn in each conversation session follows a Gamma distribution with a coefficient of variation (CV) of 0.25. Subsequently, the intervals between consecutive turns within the same session are independently sampled from another Gamma process, simulating user thinking and query formulation time. The arrival rates for inter‑session and intra‑session requests can differ, and their ratio can be varied. A higher inter‑session arrival rate relative to the intra‑session rate implies that more requests from other conversations are interleaved between two consecutive turns of the same dialogue, leading to a more dispersed request pattern and lower cache hit rates. We test two ratios, 5:1 and 10:1, to evaluate system performance under Low‑Dispersion and High‑Dispersion scenarios, respectively.

\paragraph{Metrics.} 
We use the average \textbf{Time-to-first-token (TTFT)}, \textbf{Time-per-output-token (TPOT)} as the primary metrics. To better understand the underlying mechanism, we also report cache \textbf{Hit Rates} at both the request and block granularities.

\paragraph{Baselines.}
We compare \Sys with several serving systems employing different eviction policies. 
\begin{itemize}
\item \textbf{vLLM-LRU:} A widely‑used LLM serving system built on PagedAttention~\cite{kwon2023efficient}. With prefix caching enabled, it employs the classical LRU eviction policy at the granularity of cache blocks to manage the KV cache.
\item \textbf{Max-score:} This category refers to eviction policies that compute a priority score for each cache block and evict the block with the maximum score. We implement the algorithm from~\cite{wang2025kvcache}, which estimates the reuse probability to calculate the score, and the victim block is selected in \(O(n)\) time. We use Equation \ref{eq:pw-score} to estimate the probability.
\item \textbf{Pensieve+MSA:} Pensieve~\cite{yu2025stateful} is a suffix‑caching system originally designed for single‑user multi‑turn dialogues. As it is not open‑source, we re‑implement its algorithm and further extend it with our MSA kernel to support general workloads, since Pensieve's original kernel only supports single reuse segment. Its frequency function uses an inverse‑proportional scheme, violating the order‑preserving rule (\S\ref{subsec:freq}) and resulting in linear-time eviction complexity.
\end{itemize}
For all evaluated systems, we integrate POD‑Attention~\cite{kamath2025pod}, an efficient attention kernel 
enabling a single kernel to process batches containing mixed chunked prefill and decode requests, ensuring a consistent and optimized backend across all baselines.

\begin{figure*}[tp]
    \centering
    \includegraphics[width=1.0\linewidth]{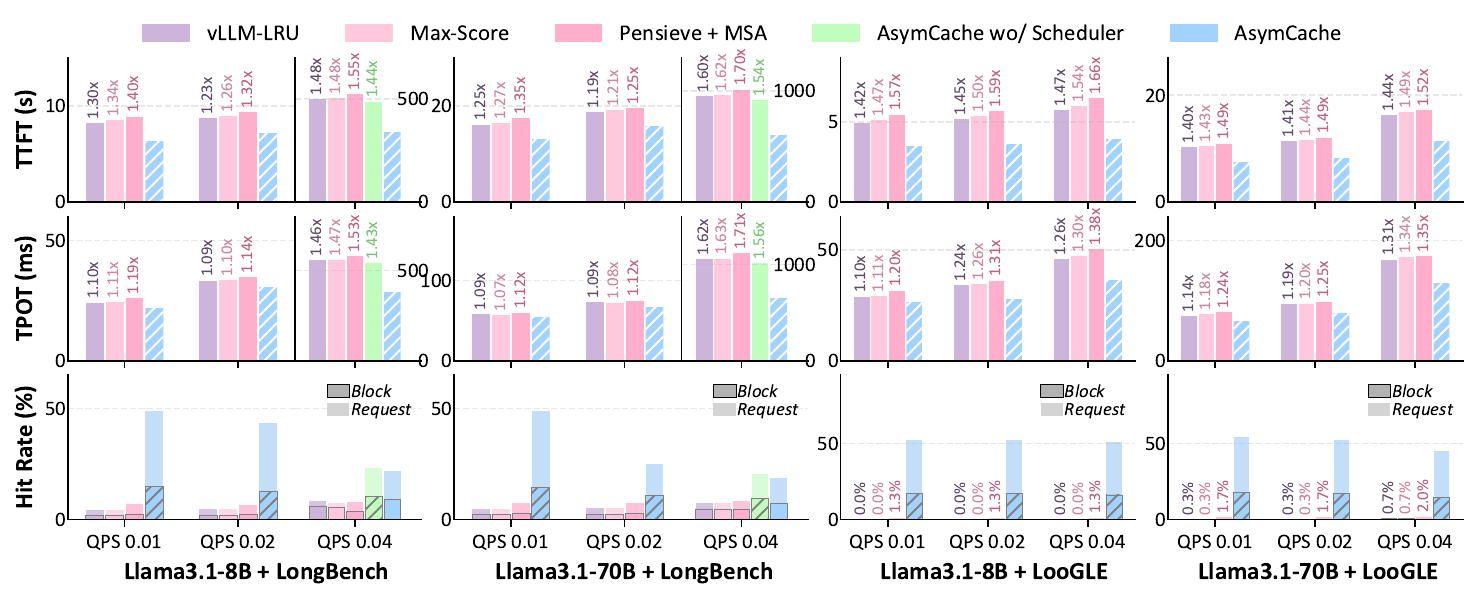}
    \caption{End-to-end results on High-Dispersion Workloads. The best performance is denoted by the bars with slashes.}
    \label{fig:e2e-f10}
\end{figure*}
\subsection{End-to-end Performance}
\label{subsec:e2e}
Figure \ref{fig:e2e-f5} and Figure \ref{fig:e2e-f10} present the end‑to‑end performance of different systems under low‑ and high‑dispersion workloads, using Llama‑3.1 8B and Llama‑3.1 70B models on the LongBench and LooGLE datasets. 
Across nearly all configurations, \Sys delivers the best TTFT and TPOT. For example, with the 70B model, \Sys reduces the TTFT over vLLM‑LRU, Max‑Score, and Pensieve+MSA by up to $1.86$×, $1.91$×, and $1.86$×, respectively; corresponding TPOT reductions reach $1.62$×, $1.63$×, and $1.71$×.

For most experiment groups (where requests do not queue), the performance gains of \Sys stem primarily from the Asymmetric Cache Block Manager. The baseline vLLM‑LRU and Max‑Score policies aim to maximize cache hit rate but ignore the asymmetric recomputation cost of tokens. Pensieve+MSA prioritizes retaining later tokens to reduce recomputation but sacrifices prefix‑hit opportunities. In contrast, \Sys holistically balances reuse probability with position‑aware recomputation savings, evicting blocks that yield the lowest expected benefit. This fundamental design leads to robust performance across varied workloads.

Low‑Dispersion workload (Figure \ref{fig:e2e-f5}) favors traditional policies: vLLM‑LRU and Max‑Score achieve high hit rates on the LongBench dataset (e.g. achieve the highest block-level cache rate 33\% on the 70B model). 
While \Sys maintains comparable or slightly lower hit rates, its ability to retain later, high‑cost tokens reduces per‑request computation, resulting in reduced latency (up to $1.36$× TTFT reduction over vLLM-LRU). 
On the LooGLE dataset, LRU and Max‑Score exhibit significant performance jitter as load increases, causing hit‑rate fluctuations ($<5\%$). \Sys mitigates this by deliberately caching latter tokens, which carry higher recomputation penalties. 
Under High‑Dispersion workload, the hit‑rate advantage of LRU and Max‑Score furthur diminishes, and their jitter becomes more pronounced (even becomes 0\% on LooGLE dataset). The hit rates of \Sys and Pensieve+MSA become higher, and \Sys consistently delivers lower latency (by up to 14\%) than Pensieve+MSA, attributing to its more efficient logarithmic‑time eviction algorithm.

Pensieve+MSA resembles \Sys in allowing suffix caching but uses an inverse‑proportional frequency function that violates the order‑preserving rule, leading to linear‑time eviction. This function poorly reflects access patterns: it prematurely evicts reusable blocks (fast initial decay) and retains stale ones, yielding no consistent block‑level hit‑rate advantage, and only marginal gains when hit rates are already high (e.g., low dispersion workload with 70B model and LooGLE, $qps=0.04$). Even when hit rates are similar, its linear overhead masks performance benefits. \Sys’s piecewise exponential design fits actual access patterns and maintains low overhead, making its latency advantage clear and consistent.

Running LongBench dataset under $qps=0.04$ induces request queuing. In this regime, we compare \Sys with \Sys wo/ Scheduler. Without the scheduler, \Sys still outperforms other baselines, confirming the intrinsic benefit of its cache manager. With the chunk scheduler enabled, \Sys dynamically selects smaller prefill chunks improving GPU utilization and allows more concurrent requests to be scheduled. This amplifies the performance advantage, yielding TTFT cuts of up to 60\% over vLLM‑LRU and 70\% over Pensieve+MSA, and TPOT cuts of up to 62\% and 71\%, respectively.

The consistent superiority of \Sys validates its core design principle: a practical eviction policy must jointly optimize the cache hit rate and the asymmetric recomputation cost. 

\subsection{Ablation Study}
\label{subsec:ablation}
Since the effects of Adaptive Chunking Scheduler have been discussed in \S\ref{subsec:e2e}, we further explore the effects of the components in Asymmetric Cache Block Manager.

\paragraph{Effects of the logarithmic-time algorithm}
As described in \S\ref{subsec:freq}, the frequency term is defined as a piecewise exponential function to enable a cache‑management algorithm with logarithmic time complexity. Here we evaluate the performance gain attributable to this design by comparing it against a baseline version that uses a straightforward linear‑time eviction policy. Table \ref{tab:on-ablation} presents the results on the Llama 3.1-8B model, showing the performance of: \Sys equipped with logarithmic/linear time eviction algorithm, and vLLM‑LRU.
Under the low‑dispersion workload, the linear‑time variant of \Sys already outperforms vLLM‑LRU, reducing TTFT by 15.2\% and TPOT by 6.6\% due to its more intelligent, reuse‑aware caching policy. Enabling the logarithmic‑time algorithm further decreases these gains by 4.7\% (TTFT) and 2.5\% (TPOT). 
A consistent trend is observed under high dispersion: the cache‑management policy itself brings initial decreases of 17.0\% in TTFT and 6.7\% in TPOT, while the logarithmic‑time optimization adds further decreases of 5.2\% and 3.2\%, respectively. 

\begin{table}[t]
\centering
\small
\begin{tabular}{c|c|ccc}
    \toprule
    Dispersion & Approach &  \textbf{TTFT}(s) & \textbf{TPOT}(ms) &\textbf{Hit}(\%) \\
    \midrule
\multirow{3}{*}{Low} & \Sys & \textbf{5.458} & \textbf{28.55} & 29.1 \\
& \Sys+$O(n)$& 5.712 & 29.25 & 28.0\\
& vLLM-LRU & 6.582 & 31.18 & 28.7 \\
\hline
\multirow{3}{*}{High} & \Sys & \textbf{7.068} & \textbf{30.51} & 12.50 \\
& \Sys+$O(n)$& 7.434  & 31.43 & 12.32\\
& vLLM-LRU &8.695 & 33.32& 2.09\\
     \bottomrule
\end{tabular}
\caption{Performance for $O(n)$-complexity algorithm. Running on Llama 3.1-8B model with LongBench ($qps=0.02$). }
\label{tab:on-ablation}
\end{table}
\begin{figure}[t]
    \centering
    \includegraphics[width=1.0\linewidth]{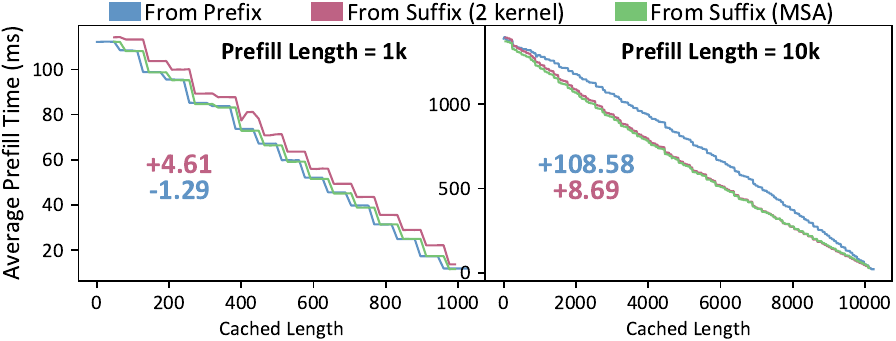}
    \caption{The MSA performance over various cached token length on model Llama 3.1-8B. The numbers indicate the average deviation between MSA and corresponding approach.}
    \label{fig:ablation-fused}
\end{figure}

\paragraph{Effects of Multi-Segment Attention kernel.}
As outlined in \S\ref{subsec:msa}, a naive implementation of MSA would require two separate attention kernel calls. To quantify the benefit of our fused MSA kernel, we compare three caching scenarios using a Llama 3.1‑8B model: (1) caching \textbf{from prefix} , (2) caching \textbf{from suffix} implemented with two kernel calls, and (3) caching \textbf{from suffix} using our single‑kernel MSA. 
Each request includes a variable number of cached tokens and ends with 128 new (uncached) tokens, simulating a user’s new input and creating a two‑segment cache layout.
For long sequences (10K tokens), caching the suffix delivers substantially lower latency than caching the prefix, as recomputation costs dominate. In this regime, the two‑kernel suffix implementation shows only a modest overhead compared to MSA (average 8.69ms). For short sequences (1K tokens), however, the absolute recomputation savings from suffix caching are small, and the fixed overhead of launching two kernels becomes significant (average 4.61ms). This makes the two‑kernel approach unattractive compared to prefix‑only caching. MSA eliminates this overhead, performing suffix‑only attention in a single kernel and achieving latency nearly identical to that of prefix‑only caching. Thus, MSA enables the practical use of suffix‑focused caching even for shorter requests, extending the applicability of our multi‑segment strategy across a wider range of sequence lengths.

\begin{figure}[t]
    \centering
    \includegraphics[width=1.0\linewidth]{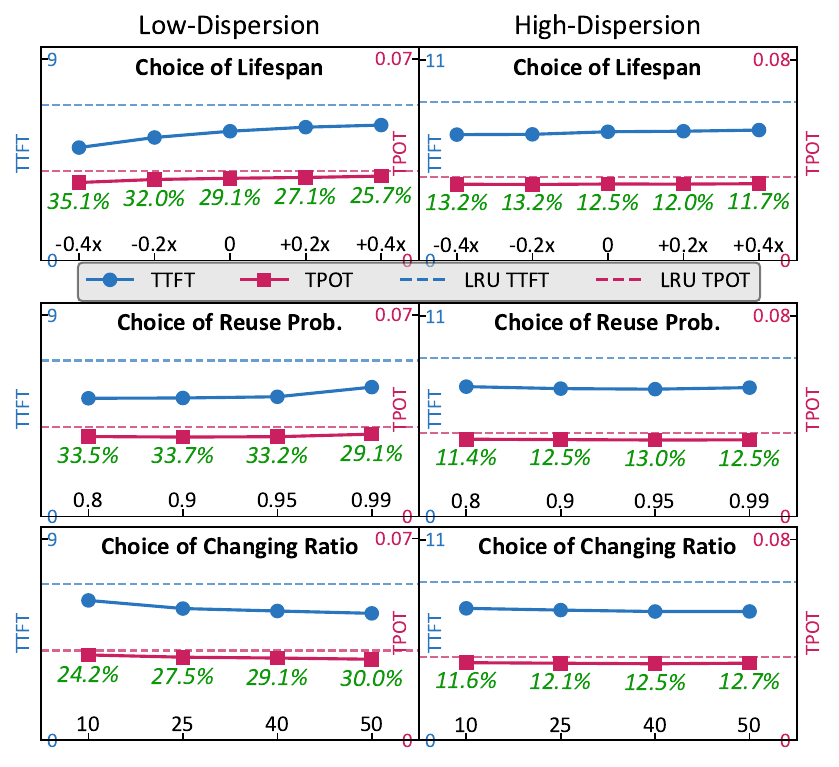}
    \caption{Hyper-parameter sensitivity study on Llama 3.1-8B with LongBench $(qps=0.02)$. The green numbers in the figure indicate the corresponding hit rates (block-level).}
    \label{fig:sen-exp}
\end{figure}

\subsection{Hyper-parameter Sensitivity}
\label{subsec:hyper-exp}
\Sys’s frequency function is defined by three parameters: lifespan (X-coordinate of the turning point), the reuse probability (Y‑coordinate of the turning point), and the slope change ratio. 
Figure~\ref{fig:sen-exp} shows that performance remains stable across a wide range of each parameter.

\paragraph{Lifespan.}
The TTFT and hit rate of \Sys vary little with lifespan. 
In our workload, the intervals between reusable KV‑cache requests are generally short. Even with a conservative P99‑based parameter choice, a reduced lifespan still covers the vast majority of requests. 
In all values tested, \Sys consistently outperforms vLLM‑LRU.
It is worth noting that while an excessively short lifespan may prematurely discard reusable blocks and an overly long one may retain low‑utility data, the performance penalty remains marginal within a reasonable range.

\paragraph{Reuse Probability}
The shape of the frequency function is stable as long as the value is sufficiently high, making the hit rate and the latency steady. Setting it too low shifts the policy toward recomputation‑cost-optimal eviction, which may hurt performance in reuse‑heavy workloads.
We find that setting this probability within a moderate range (e.g., between 0.3 and 0.7) yields nearly optimal performance. This insensitivity stems from the fact that the eviction policy operates on relative priorities among cache blocks rather than absolute probability thresholds, making the system adaptable to different probability calibrations.

\paragraph{Slope change ratio}
The performance of \Sys drops when the ratio is set to its minimum value (10), where the blocks decay too slowly and the low‐utility blocks persist. Despite obvious sensitivity at a lower value, \Sys maintains stable performance across all other tested values, demonstrating robustness within a practical operating range.
The eviction algorithm is sufficiently efficient in managing the priority updates regardless of the decay rate, ensuring consistent performance across different slope choices

\begin{figure}[t]
    \centering
    \includegraphics[width=1.0\linewidth]{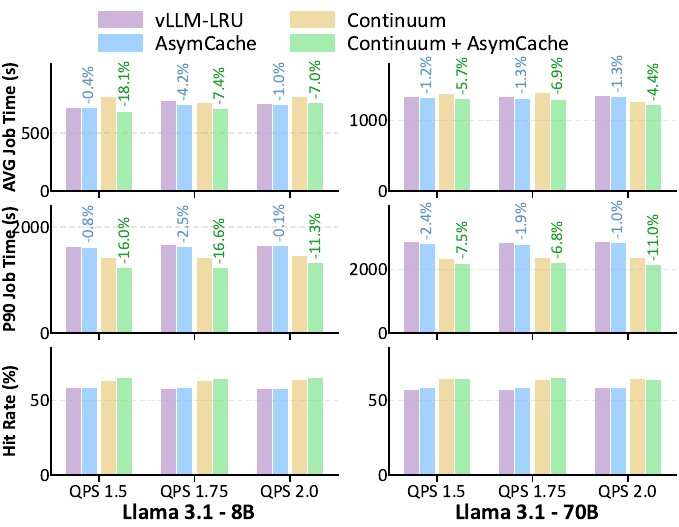}
    \caption{Results on BFCL Dataset, an agentic workload.}
    \label{fig:agent-exp}
\end{figure}

\subsection{Integration with Agentic System}
We evaluate \Sys in an agent serving scenario using the Berkeley Function Calling Leaderboard (BFCL) v4~\cite{patil2025berkeley}, following two prior agent serving systems Autellix~\cite{luo2025autellix} and Continuum~\cite{li2025continuum}. We use the BFCLv4 Web Search part and pre‑generate outputs with GPT‑5.1 to ensure reproducibility and fair comparison across test groups. Experiments are conducted on both Llama 3.1‑8B and Llama 3.1‑70B across varying QPS. The results are reported in Figure~\ref{fig:agent-exp}.

In agent workloads, the intervals between consecutive turns are shorter and more predictable than in human‑driven conversations, as described in \S\ref{subsec:application}. This leads to a low‑dispersion pattern where the overall workload behavior closely resembles the low‑dispersion multi‑turn scenario: cache hit rates between vllm-LRU and \Sys are nearly identical, or \Sys sometimes even worse. However, since \Sys preferentially retains blocks at later sequence positions, which incur substantially higher recomputation costs, it achieves average job latency reductions of 0.4\% to 4.2\% compared to vllm-LRU.
Continuum introduces a TTL‑based KV cache management for agent workloads: it estimates the execution time of each tool call, assigns a time‑to‑live to the corresponding KV cache, pins the cache in GPU memory during the TTL window, and prioritizes requests with pinned caches in the scheduler. This request‑level strategy already outperforms conventional approaches (particularly in P90 tail job latency metric). 

We also implement \Sys on top of Continuum without modifying its core scheduling logic. While Continuum manages KV cache at request granularity, \Sys further optimizes within each request by prioritizing lower‑value cache blocks to evict. 
Under nearly identical cache hit rates, the combined system reduces average latency by 4.4\% - 18.1\% compared to Continuum alone, and P90 tail job latency achieves 6.8\% - 16.6\% reduction as well. 
These results demonstrate the generality and composability of \Sys as a block‑level optimization that complements existing request‑oriented agent serving systems.

\label{subsec:agent-exp}

%% file: chapters/conclusion.tex
\section{Discussion and Related Works}
\paragraph{Systems with KV-Cache reusing.}
\indent
Recent research has focused on leveraging KV-cache reuse to accelerate inference. Systems like SGLang~\cite{zheng2024sglang} and vLLM~\cite{kwon2023efficient} enable prefix caching. The former at the request granularity using a RadixTree algorithm, and the latter at the cache block level via PagedAttention. However, both rely on LRU eviction, which ignores the asymmetric recomputation cost of tokens and may underperform in long-context scenarios compared to \Sys.
Other works explore caching KV-cache on on-chip memory for shared reuse across users, similar to \Sys. PromptCache~\cite{gim2024prompt} requires a pre-defined schema to specify reusable segments explicitly. RAGCache~\cite{jin2025ragcache} targets retrieval-augmented generation by managing cache with a tree structure. ChunkAttention~\cite{ye-etal-2024-chunkattention} achieves efficient KV-cache management by splitting contexts into chunks and building a prefix tree, yet it remains focused on prefix-only caching. These approaches are often specialized to particular scenarios.
CachedAttention~\cite{gao2024cost} and Pensieve~\cite{yu2025stateful} are specialized for multi-turn dialogue but do not support cross-user cache sharing. CachedAttention introduces a scheduler-aware eviction policy that considers the task queue to determine eviction priority, an idea that is orthogonal and potentially complementary to \Sys. Pensieve prioritizes caching tokens in later positions to reduce recomputation costs. However, as our experiments in \S\ref{subsec:e2e} show, its inability to cache high-reuse prefixes and its linear-time eviction complexity limit its flexibility and general applicability relative to \Sys.

\paragraph{KV‑Cache offloading and hierarchical storage.}
While \Sys currently stores KV cache solely in GPU memory, several existing systems explore offloading it to DRAM, local disk, or even remote nodes. Pensieve~\cite{yu2025stateful}, RAGCache~\cite{jin2025ragcache}, and APIServe~\cite{abhyankar2024apiserve} incorporate DRAM as an additional tier in a multi‑level cache hierarchy. HCache~\cite{gao2025fast}, CachedAttention~\cite{gao2024cost}, Cake~\cite{jin2024compute}, AttentionStore~\cite{gao2024attentionstore} and HotPrefix~\cite{li2025hotprefix} further extend this hierarchy to local disk, while CacheGen~\cite{liu2024cachegen} utilizes remote machines and investigates compression techniques for transmission. Many of these works also explore overlapping data transfer with computation when moving KV cache from a lower‑level storage back to the GPU, a technique employed, for instance, by HCache.
Although \Sys does not yet support a multi‑level storage hierarchy, its MSA mechanism and eviction algorithm are compatible with such an architecture; evicted blocks could simply be written to a lower storage tier. However, further optimization opportunities remain. For example, the latency of loading a block from a lower tier depends primarily on its size rather than its position, unlike recomputation costs. This suggests that eviction policies could differ across storage tiers, a direction we plan to explore in future work.


\paragraph{Agentic Serving Systems.}
Recent work recognizes that LLM agents exhibit structured execution patterns distinct from conventional chatbots. Parrot~\cite{lin2024parrot} exposes application‑level DAGs via Semantic Variables to enable cross‑request optimization. Plan Caching~\cite{zhang2025cost} caches reusable plan templates extracted from completed agent tasks, reducing cost while preserving accuracy.
At the storage layer, InferCept~\cite{abhyankar2024infercept} manages KV Cache with min-waste interception policy. KVFlow~\cite{pan2026kvflow} replaces LRU with an Agent Step Graph that prioritizes blocks by their estimated steps to execution. Continuum~\cite{li2025continuum} employs TTL to pin KV cache during tool call stalls. Tokencake~\cite{bian2025tokencake} uses agent‑aware spatial partitioning and time‑based offloading to mitigate memory contention and idle waiting. At the scheduling layer, Autellix~\cite{luo2025autellix} treats programs as first‑class citizens and preempts based on completed calls, ThunderAgent~\cite{kang2026thunderagent} abstracts agent workflows as Programs to unify heterogeneous resources.
These systems primarily optimize at request or program granularity: scheduling across requests, preserving entire KV caches during tool calls, or caching plan templates. In contrast, \Sys operates at per‑request, per‑block granularity, optimizing eviction within a single request based on asymmetric recomputation costs. Our approach is orthogonal and can be integrated into existing agent serving systems as a complementary component (\S\ref{subsec:agent-exp}).
Tighter co‑design between block‑level eviction and program‑level policies remains future work.

\section{Conclusion}
We present \Sys, an LLM serving system built upon a novel Multi-Segment Attention (MSA) mechanism that enables flexible management of the KV cache. Different from existing caching policies, \Sys consciously exploits the asymmetric recomputation costs of cache blocks while carefully balancing reuse probability against caching benefit, prioritizing KV cache eviction
based on each block’s marginal contribution to expected attention
latency. This leads to a general, efficient, and flexible cache management algorithm, which is further supported by chunk scheduler.
Experiments demonstrate that \Sys achieves an average TTFT improvement of 1.90–2.03× and an average TPOT improvement of 1.62–1.71× over existing approaches. The benefits extend to agent serving workloads, where \Sys integrates seamlessly into existing systems such as Continuum, reducing average job latency by up to 18.1\%.

%% file: chapters/appendix.tex
\clearpage
\section{Appendix}

\subsection{Properties of Order-Preserving Rule}
We now proof that only exponential function can satisfy the order-preserving rule proposed in Section \ref{subsec:freq}. 
\begin{lemma}\label{lem:exp}
    Let $f: \mathbb{R} \to \mathbb{R}$ be a continuous, non-negative, and non-constant function. Furthermore, assume $f$ satisfies the functional equation:
    \[
    \frac{f(x+y)}{f(x)} = g(y) \quad \text{for all } x, y \in \mathbb{R},
    \]
    where $g: \mathbb{R} \to \mathbb{R}$ is some function. Then $f$ is necessarily of the exponential form $f(x) = C_1 \cdot e^{C_2 x}$ for some constants $C_1 > 0$ and $C_2 \in \mathbb{R}$.
\end{lemma}

\begin{proof}
    Since $f$ is non-negative, non-constant, and satisfies the above equation, we have $f(x) > 0$ for all $x$. (If $f(x_0)=0$ for some $x_0$, then the equation implies $f(x)=0$ for all $x$, contradicting non-constancy). Define $h(x) = f(x) / f(0)$. Then $h(0) = 1$, $h$ is continuous, and from the given equation with $x=0$, we find $g(y) = f(y)/f(0) = h(y)$. Thus, the functional equation becomes:
    \[
    \frac{f(x+y)}{f(x)} = h(y) \quad \text{for all } x, y.
    \]
    Substituting the definition of $h$, we have:
    \[
    \frac{f(x+y)}{f(x)} = \frac{f(y)}{f(0)} \quad \text{or equivalently,} \quad \frac{f(x+y)}{f(0)} = \frac{f(x)}{f(0)} \cdot \frac{f(y)}{f(0)}.
    \]
    This implies that $h$ satisfies Cauchy's multiplicative equation:
    \[
    h(x+y) = h(x) h(y) \quad \text{for all } x, y \in \mathbb{R}.
    \]
    Since $h$ is continuous, non-negative (as $f$ is), and $h(0)=1$, the only non-trivial solutions are exponential functions: $h(x) = e^{\lambda x}$ for some $\lambda \in \mathbb{R}$. Recalling that $f(x) = f(0) \cdot h(x)$, we conclude:
    \[
    f(x) = C_1 \cdot e^{C_2 x},
    \]
    where $C_1 = f(0) > 0$ and $C_2 = \lambda \in \mathbb{R}$.
\end{proof}

\begin{theorem}
    Let $f: \mathbb{R} \to \mathbb{R}$ be a continuous, non-negative, and non-constant function. If $f$ satisfies the following \emph{order-preserving rule}: for all real numbers $x, y, a, b, s$ with $a, b \ge 0$ and not both zero, the expression $a f(x) - b f(y)$ has the same sign (non-negative or non-positive) as $a f(x+s) - b f(y+s)$, then $f$ must be an exponential function of the form $f(x) = C_1 \cdot e^{C_2 x}$.
\end{theorem}

\begin{proof}
    The order-preserving rule implies that the function preserves the ordering of weighted sums under translation. We first show that this condition forces the ratio $f(x+y)/f(x)$ to be independent of $x$.

    \textbf{Step 1: Deriving the functional equation.}
    Since $f$ is non-constant, and non-negative, and the condition holds for all $a,b \ge 0$, we can choose specific weights to probe its consequences. Fix arbitrary $x, y, s \in \mathbb{R}$. Consider the difference $f(x) - t f(y)$ for some $t \ge 0$. The given rule states that for any $s$,
    \[
    \operatorname{sgn}\big( f(x) - t f(y) \big) = \operatorname{sgn}\big( f(x+s) - t f(y+s) \big),
    \]
    where $\operatorname{sgn}(\cdot)$ denotes the sign (or zero). This means the map $s \mapsto [f(x+s) - t f(y+s)]$ does not change sign as $s$ varies. In particular, if we can find a $t$ such that $f(x) - t f(y) = 0$, then we must also have $f(x+s) - t f(y+s) = 0$ for all $s$.

    Let $t_0 = f(x) / f(y)$. Since $f$ is positive (as argued in Lemma \ref{lem:exp}), $t_0$ is well-defined and non-negative. By construction, $f(x) - t_0 f(y) = 0$. Therefore, the rule forces:
    \[
    f(x+s) - t_0 f(y+s) = 0 \quad \text{for all } s \in \mathbb{R}.
    \]
    Substituting $t_0 = f(x)/f(y)$, we obtain:
    \[
    f(x+s) = \frac{f(x)}{f(y)} f(y+s) \quad \text{for all } s.
    \]
    Equivalently, for all $x, y, s$,
    \[
    \frac{f(x+s)}{f(x)} = \frac{f(y+s)}{f(y)}.
    \]
    This shows that the ratio $\frac{f(x+s)}{f(x)}$ is independent of $x$; it depends only on the increment $s$. Hence, there exists a function $g: \mathbb{R} \to \mathbb{R}$ such that:
    \[
    \frac{f(x+s)}{f(x)} = g(s) \quad \text{for all } x, s \in \mathbb{R}.
    \]
    This is precisely the functional equation required in Lemma \ref{lem:exp}.

    \textbf{Step 2: Applying Lemma \ref{lem:exp}.}
    The function $f$ satisfies all hypotheses of Lemma \ref{lem:exp}: it is continuous, non-negative, non-constant, and we have just shown it satisfies $f(x+s)/f(x) = g(s)$. Therefore, by Lemma \ref{lem:exp}, $f$ must be of the exponential form $f(x) = C_1 \cdot e^{C_2 x}$ for some constants $C_1 > 0$ and $C_2 \in \mathbb{R}$.
\end{proof}